\documentclass[letter,11pt]{article}
\usepackage[letterpaper,text={6in,8.5in}]{geometry}

\usepackage{authblk}
\usepackage{hyperref}
\usepackage{amsfonts}
\usepackage{amsmath}
\usepackage{amsthm}
\usepackage{graphicx}
\usepackage{color}

\usepackage{algorithm}
\usepackage{algpseudocode}
\algnewcommand{\IIf}[1]{\State\algorithmicif\ #1\ \algorithmicthen}
\algnewcommand{\EndIIf}{\unskip\ \algorithmicend\ \algorithmicif}

\theoremstyle{plain}
\newtheorem{theorem}{Theorem}[section]
\newtheorem{lemma}[theorem]{Lemma}

\newtheorem{corollary}[theorem]{Corollary}
\theoremstyle{definition}

\newcommand{\fix}[1]{\textcolor{red}{\textbf{#1}}}

\newcommand{\Xomit}[1]{}

\begin{document}

\pagenumbering{gobble}
\begin{titlepage}

{ 
    \title{Idle Ants Have a Role
	    {\let\thefootnote\relax\footnote{This research was supported by the Israel Science Foundation (grant 1386/11).}}
    }

    \author[1]{Yehuda Afek}
    \author[2]{Deborah M. Gordon}
	\author[1]{Moshe Sulamy}
    \affil[1]{Tel-Aviv University}
    \affil[2]{Stanford University}
    \date{}
    \maketitle
}

    \begin{abstract}
		Using elementary distributed computing techniques we suggest an explanation for two unexplained phenomena in regards to ant colonies,
		(a) a substantial amount of ants in an ant colony are idle, and (b) the observed low survivability of new ant colonies in nature.
		Ant colonies employ task allocation, in which ants progress from one task to the other, to meet changing demands introduced by the environment.
		Extending the biological task allocation model given in [Pacala, Gordon and Godfray 1996]
		we present a distributed algorithm which mimics the mechanism ants use to solve task allocation efficiently in nature.
		Analyzing the time complexity of the algorithm reveals an exponential gap on the time it takes an ant colony to satisfy a certain work demand with and without idle ants.
		We provide an $O(\ln n)$ upper bound when a constant fraction of the colony are idle ants, and a contrasting lower bound of $\Omega(n)$ when there are no idle ants, where $n$ is the total number of ants in the colony.
    \end{abstract}

    \hfill
	\begin{center}
		
	\end{center}

\end{titlepage}
\pagenumbering{arabic}

\section{Introduction}
\label{section_intro}

Biological and distributed systems have a lot in common, and
the study of one system may inspire new observations on the other
system in a reciprocal manner \cite{AfekZiv}.
Both fields have a strong distributed aspect, involving many entities that communicate
locally to solve a global problem. 
In \cite{AfekZiv} a collection of cells solve the Maximal Independent Set (MIS) problem
through simple communication between neighboring cells;
the MIS cells are those that grow a bristle on the fruit-fly forehand. 
In \cite{Feinerman2012Collaborative,ANTS1,ANTS2,ANTS3,ANTS4,ANTS5,ANTS6} the ants in a colony collaboratively search and find a food item placed
at some arbitrary distance from the nest.  In \cite{CornejoDLN14} the
ants in a colony switch tasks according to an indication about
the work loads in order to achieve an optimal division of labor.
Here, inspired by Gordon \cite{GordonInteraction} and Pacala et al. \cite{Pacala},
we extend the biological task allocation model of \cite{Pacala}
to consider the communication aspects and the more local
and distributed nature of the task allocation problem in ant colonies,
and utilize elementary distributed computing techniques to provide possible explanations for phenomena observed by biologists.

According to \cite{GordonBook} at each point of time the ants in an ant colony are partitioned into several tasks.
Examples of tasks include: brood tending, cleaning, patrolling, and foraging for food.
Once an ant is in a task, it can be in one of three states: working, not-needed, or idle.
A \emph{working} ant is successfully engaged in its current task,
a \emph{not-needed} ant is engaged in its current task but has no work to do,
while an \emph{idle} ant is not actively doing anything.
To achieve optimal task allocation, the ants communicate locally with each other
to decide whether to switch from one task to another.
A not-needed ant in task $t_1$ that senses, through the accumulation of one-to-one contacts,
a sufficient number of working ants in task $t_2$, switches tasks from $t_1$ to $t_2$.
The ant is thus "recruited" to task $t_2$.
This has the outcome that an ant switches tasks if it is not needed in its current task
and there is high demand from the following task.
Notice that an ant has no knowledge of the available tasks, and what it does next is influenced by interactions \cite{AdditionalGordon}.
Thus, ants switch tasks \emph{only} through encounters, and not independently.
This distributed recruitment algorithm dictates the changes among the tasks.


According to biologists, task allocation in some species involves an irreversible sequence of tasks,
where ants can only promote in their task switching \cite{GordonBook,GordonInteraction,gordon1989dynamics}.
An individual advances in only one way through a series of tasks, and does not move backwards in this sequence.
For example, harvester ants proceed through tasks as:
Nursing $\rightarrow$ Cleaning $\rightarrow$ Patrolling $\rightarrow$ Foraging.
We incorporate this in our model as well, allowing ants to switch tasks only one-way,
by interacting with working (or idle) ants of the subsequent task.

In recent years, biological observations \cite{lazybones,WhyLazy,LazyInNature} discovered
that, in many species of ants, a substantial amount (20\%-50\%) of the ants in the colony
are idle, even when there is work to be done.
We suggest a possible explanation for this by showing that if a constant fraction of the colony are idle ants,
the time it takes to converge into an appropriate allocation of ants to tasks is improved exponentially.
We also consider how our model may help to explain another observation,
that in some species the survival rate of young nests is low.
For example, in a population of harvester ants, only 10\% of new colonies survive the first year \cite{LowSurvive,GordonBook}.

We consider the \emph{task allocation} problem, in which a set of ants start in some initial assignment to tasks,
and the goal is to switch tasks until a demand assignment is reached.
The \emph{demand} specifies a lower bound on the number of ants required in each task,
reflecting the workload imposed by external environmental conditions,
such as the amount of food to be obtained, debris to be cleared and so on.

Existing observations and experiments suggest that the rate of encounter is crucial in the decision of an ant on which task to perform next \cite{AdditionalGordon,InterRate,InterRateRegulation}.
In our model, called ANTTA (ANTs Task Allocation), we support \emph{rate of interaction}, determining the number of interactions for each ant before its decision point.
We assume the rate is constant throughout our protocols.

The ANTTA model assumes a synchronous round based system where,
in each round, an ant decides on its task in the next round according to multiple interactions with other ants in the colony.
The interacting ants exchange their local information,
i.e., which task each belongs to, and their state.
We say an ant is \emph{recruited} if it switches to another task,
which occurs if it is not needed in its current task and encounters a sufficient number of working or idle ants in another task.

\Xomit{
	According to biologists \cite{GordonBook,GordonInteraction,gordon1989dynamics},
	ants can only promote in their task switching.
	Thus, in ANTTA we extend the previous models and present a task allocation model
	where ants can only switch tasks one-way, by interacting with successful ants of the following task.
}

\subsection{Contributions}

The main contribution of this paper is the usage of elementary computing and
distributed computing techniques  to suggest explanations of biological phenomena.

Biologists observed two phenomena in nature which are, to the best of our knowledge, as of yet unexplained.
First, a constant fraction (20\%-50\%) of the ants in the colony are idle, even when there is work to be done.
Second, the survival rate of new nests is very low, having only 10\% of new nests survive through the first year \cite{LowSurvive,GordonBook}.
We argue that our model and results provide possible explanations for each of these phenomena.

To explains the first phenomenon we prove a lower bound of $\Omega(n)$ rounds on ANTTA when there are no idle ants in the colony.
We then present a recruitment algorithm that solves ANTTA in $O(\ln{n})$ rounds when a constant fraction of the colony are idle ants,
thus presenting an exponential gap on the task allocation time complexity with or without idle ants, providing a possible explanation for their presence.

As to the second phenomenon, biologists observed that as nests mature, their percentage of idle ants increases \cite{GordonBook,gordon1989dynamics}.
Young nests contain very few or no idle ants, and thus, by our hypothesis, 
have a slow response time to sudden changes in demands for different tasks,
providing a possible explanation for their low survivability.
When encountering such demands, the nest may be incentivized to turn more ants into idle, in order to improve the response time of task allocation;
however, the idle ants need to be distributed among the tasks to which they recruit (here we assume that each idle ant recruits to a specific task only).
We show that this process takes exponentially longer than ANTTA with idle ants, and thus the reaction time of young nests to sudden changes in demand remains slow ($\Omega(n)$).
We therefore extend our model to discuss the process by which idle ants are distributed among the tasks,
where the goal is to reach a state in which a constant fraction of idle ants are in each task, enabling the fast solution of task allocation.
We present a lower bound of $\Omega(n)$ rounds required to distribute the idle ants.
Thus, when a new demand is encountered for the first time,
young nests still have to solve the idle ants distribution,
and are thus still slow to respond, providing further explanation for their low survivability.

\section{ANTTA Model}
\label{section_model}

We consider a synchronous system that progresses in rounds.
Let $n$ denote the number of ants,
$t$ the number of different tasks that are available,
and $T_a^r \in \{1,\dots,t\}$ the task an ant $a$ belongs to in round $r$.
Denote by $X^r=\{x_1^r,\dots,x_t^r\}$, where $x_i^r$ is the number of ants in task $i$ at round $r$,
the \emph{assignment} vector at round $r$.
We denote the \emph{demand} by $D=\{d_1,\dots,d_t\}$,
where $d_i$ is the required number of ants for task $i$.
A demand $D$ that satisfies $\sum_i d_i \leq \gamma n$ is called \emph{$\gamma$-demand}.
For simplicity we assume that $d_i>0$ and $x_i^r>0$ for each $i,r$,
and that every demand is a $\gamma$-demand for some $\gamma \in \left[ 0,1 \right]$.

Denote by $S_a^r \in \{I,W,N\}$ the state of ant $a$ in round $r$,
which indicates whether an ant is Idle (I),
Working and engaged in its current task (W), or Not-needed (N).
Let $B = \{\beta_1,\dots,\beta_t\}$ denote the \emph{idle factors}, such that there are $\beta_i n$ Idle ants in task $i$,
where $0 \leq \beta_i \leq 1$, thus $\beta_i$ is the percentage of idle ants in task $i$, out of the total number of ants.
Other ants never become idle, and idle ants never change to a different state.
Idle ants do not engage in their task or initiate interactions,
and thus do not count towards the \emph{assignment} or \emph{demand},
though they are interacted with.

Let $\alpha = 1-\sum_i \beta_i$,
we say that the demand $D$ is \emph{satisfiable} by the assignment $X^0$ and the idle factors
if it is an $\alpha$-demand and it holds that $\forall i$: $d_i \leq \sum\limits_{j \leq i} (x_j^0 - \beta_j n)$.

Denote by $R$ the \emph{interaction rate}, the number of interaction an ant performs each round; we assume $R$ is constant.
At the beginning of each round $r$, $\min(d_i, x_i^r)$ ants are assigned $W$ in each task $i$,
arbitrarily chosen from all non-Idle ants in task $i$, and the rest non-idle are assigned $N$.
Each non-idle ant $a$ then performs $R$ interactions with other ants.
In each interaction ant $a$ selects, uniformly at random, one ant out of the entire nest to interact with.
Denote by $b_1, \dots, b_R$ the set of ants $a$ interacted.

In \emph{ANTTA}, task switching is only done one-way, as observed in nature \cite{GordonBook}.
Notice that ants do not switch tasks independently, only through interactions \cite{AdditionalGordon}.
Thus, if the task of ant $a$ directly precedes that of ant $b_i$ for some $i$,
it may decide whether to switch to the task of $b_i$, or remain in its current task,
i.e., if $T_a^r+1 \in \{ T_{b_1}^r, \dots, T_{b_R}^r \}$ then it holds that $T_a^{r+1} \in \{T_a^r$, $T_a^r+1\}$,
otherwise $T_a^{r+1}=T_a^r$.

In the \emph{task allocation problem},
an adversary determines the \emph{initial} assignment $X^0$ and a \emph{satisfiable} demand $D$.
The goal of the protocol is to reach an assignment which matches the demand.
A protocol terminates successfully in the first round $f$
in which the assignment meets (or exceeds) the demand,
i.e., when $\forall i$: $x_i^f \geq d_i$.

\section{Recruitment Algorithm}
\label{section_algorithm}

Here we present a task allocation algorithm whose expected time complexity when there is a constant fraction of idle ants is $O(\ln n)$ in the worst case.
For completeness, we additionally show that the time complexity of the algorithm
when there are no idle ants is $O(n \ln n)$.
Our algorithm is based on the biological model of ants task allocation and interactions,
as described in \cite{GordonBook,Pacala}.

\subsection{Algorithm}

The Recruitment Algorithm is based on the idea that,
if a not-needed ant encounters a sufficient number of working ants in the next task,
there is a surplus of ants in its current task and the demand is probably better met by switching to the next task.

Each round, each ant interacts with $R$ other ants chosen at random from all the ants.
If a not-needed ant interacts with a sufficient number of
idle or working ants from its subsequent task,
it switches to that task, in all other cases the ant remains in its current task.
This switch is likely since the probability to interact with a working ant
from the subsequent task is higher when the demand in the next task
has not yet been satisfied and, as we claim later,
recruitment by idle ants dramatically improves the time it takes to meet certain demands.
This in effect means working and idle ants "recruit" other, not-needed ants
to their task, as observed in ant colonies \cite{Pacala,GordonBook,AdditionalGordon}.
The sufficient number of interactions with working or idle ants
required in order for a not-needed ant to switch tasks
is determined by a threshold parameter $th$, where $1 \leq th \leq R$.

In each round $r$, each non-idle ant performs the \texttt{Decide} method (see Algorithm \ref{alg_recruit}).  
The \texttt{Decide} method for ant $a$
receives as arguments $T_a^r, S_a^r$, the current task and state of $a$,
and $\{ T_{b_1}^r, S_{b_1}^r, \dots, T_{b_R}^r, S_{b_R}^r \}$, the current task and state of all $R$ ants with which $a$ interacts with in round $r$.
Recall that idle ants do not initiate interactions and thus their \texttt{Decide} method is not invoked.


\begin{algorithm}
	\caption{Recruitment Algorithm; \texttt{Decide} (performed each round by non-idle ants).}
	\label{alg_recruit}
	\begin{algorithmic}[1]
		\Function{Decide}{$T_a^r,S_a^r, T_{b_1}^r, S_{b_1}^r, \dots, T_{b_R}^r, S_{b_R}^r$}
		\If{$(S_a^r = W)$}							\Comment{Current working}
		\State Return							\Comment{No change}
		\Else										\Comment($a$ is not-needed)
		\State Set $C := \left|\{ k | S_{b_k}^r \neq N, T_{b_k}^r = T_a^r + 1  \}\right|$	\Comment{Idle and working of next task}
		\If{$(C \geq th)$}						\Comment{Encountered more than the threshold}
		\State Set $T_a^{r+1} :=T _a^r + 1$		\Comment{Switch to next task}
		\EndIf
		\EndIf
		\EndFunction
	\end{algorithmic}
\end{algorithm}

We now analyze the time complexity of Algorithm \ref{alg_recruit}, with and without idle ants in the colony.
We show that, when idle ants in each task constitute a constant fraction of the total number of ants,
the expected time complexity for the algorithm is $O(\ln{n})$ rounds in the worst case,
presenting an exponential gap from the lower bound (presented below in Section \ref{section_bounds})
and providing a possible explanation for the existence of idle ants.
For completeness, we show that without idle ants, the expected time complexity for the algorithm
is $O(n \ln n)$ in the worst case.

\subsection{Analysis}

\subsubsection{No Idle Ants}


\begin{theorem}
	\label{theorem_alg_runtime}
	When there are no idle ants, the expected time complexity of Algorithm \ref{alg_recruit}
	is $\Theta(n \ln{n})$ rounds in the worst case.
\end{theorem}

\begin{proof}
	
	We prove the theorem in two steps.
	First, we show that for each assignment and demand, the algorithm terminates in at most an expected $O(n \ln{n})$ rounds.
	We then show an assignment and demand in which the algorithm takes at least an expected $\Omega(n \ln{n})$ rounds,
	proving the bound is tight and the expected time complexity is $\Theta(n \ln{n})$ rounds.
	
	\begin{lemma}\label{lemma_alg_upper}
		When there are no idle ants, the expected runtime of Algorithm \ref{alg_recruit} is at most $O(n \ln n)$ rounds.
	\end{lemma}
	
	\begin{proof}
		For the upper bound we assume $th=R$; lower values of $th$ can only cause the demand to be met faster,
		thus the bound we show holds for any $1 \leq th \leq R$.
		Let $X^0, D$ be an initial assignment and demand.
		If the assignment does not meet the demand, then there must be some task $i<t$ which exceeds the demand, such that $x_i^0 > d_i$.
		We show that after $n \ln n$ rounds, the expected number of not-needed ants remaining in task $i$ is at most one.
		In the worst case, this happens sequentially one task after the other, 
		and we get that after $t \cdot n \ln{n}$ rounds	no task (except the last) exceeds its demand, thus all tasks meet the demand, and the algorithm terminates successfully.
		Since the number of tasks $t$ is constant, the expected time complexity
		is at most $O(n \ln{n})$ rounds.
		
		Let $k_i^r = x_i^r - d_i$ denote the number of not-needed ants in task $i$ at round $r$.
		The number of recruiting (i.e., working) ants in task $i+1$ is at least $1$,
		thus for each not-needed ant in $i$,
		the probability to interact with a sufficient number of working ants in $i+1$,
		and thus be recruited to task $i+1$, is at least $\frac{1}{n^R}$.
		The probability of an ant to \emph{not} be recruited is thus at most $1-\frac{1}{n^R}$.
		Let $x$ be some number of rounds, the probability for an ant to remain in task $i$ after $x$ rounds is at most $(1-\frac{1}{n^R})^x$.
		The expected number of ants remaining in task $i$ after $x$ rounds is thus at most $k_i^r (1 - \frac{1}{n^R})^x$.
		Assigning $x = n \ln n$, after $n \ln n$ rounds the expected number of ants remaining in task $i$ is thus at most:
		$$k_i^r \cdot (1 - \frac{1}{n^R})^{n \ln n} \leq k_i^r \cdot \frac{1}{\mathrm{e}}^{\ln n} = \frac{k_i^r}{n} \leq 1 $$
		
		Thus, after $n \ln n$ rounds, the expected number of not-needed ants remaining in task $i$
		is at most one
		(which switches after an expected $n$ rounds in the worst case),
		and the overall upper bound is $O(n \ln n)$.
	\end{proof}
	
	\begin{lemma}
		When there are no idle ants, the expected time complexity of Algorithm \ref{alg_recruit} is $\Omega(n \ln n)$ rounds in the worst case.
	\end{lemma}
	
	\begin{proof}
		We will now show an initial assignment and demand in which the algorithm terminates in \emph{at least} an expected $\Omega(n \ln{n})$ rounds.
		For the lower bound we assume $th=1$; larger values of $th$ can only cause the demand to be met slower,	thus the bound we show holds for any $1 \leq th \leq R$.
		Let us define the following assignment and demand vectors for $t=3$ tasks and $n$ ants:
		$$ X^0 = \{ n-2, 1, 1 \} ; D = \{ 1, 1, n-2 \} $$
		
		We show that after $\frac{(n-1) \ln(n-3)}{R} = O(n \ln n)$ rounds,
		the expected number of not-needed ants in task $1$ is \emph{at least} one,
		thus the algorithm does not terminate and thus the overall runtime is $\Omega(n \ln n)$.
		
		The number of working ants in tasks $1$ and $2$ is always $1$,
		thus in each round $r$, the number of not-needed ants in task $1$ is $k_1^r = x_1^r - 1$,
		and each not-needed ant in task $1$ has a probability of $(1-\frac{1}{n})^R$
		to \emph{remain} in task $1$ and not be recruited to task $2$.
		Let $x$ be some number of rounds,
		the probability for a not-needed ant in task $1$ to remain
		for \emph{at least} $x$ rounds is $(1 - \frac{1}{n})^{Rx}$.
		The expected number of not-needed ants remaining in task $1$
		after $x$ rounds is thus $(n-3) (1 - \frac{1}{n})^{Rx}$.
		Assigning $x = \frac{(n-1) \ln(n-3)}{R}$, the expected number of unsuccessful ants remaining in task $1$ is:
		$$ (n-3) (1 - \frac{1}{n})^{(n-1) \ln (n-3)} > (n-3) \cdot \frac{1}{\mathrm{e}}^{\ln (n-3)} = \frac{n-3}{n-3} = 1 $$
		
		
		Thus, after $O(n \ln n)$ rounds, the expected number of not-needed ants in task $1$
		is at least one,
		meaning the demand is not satisfied, and the overall lower bound is $\Omega(n \ln n)$.
	\end{proof}
	
	We have shown an upper bound and lower bound of $O(n \ln n)$,
	and thus the expected time complexity of Algorithm \ref{alg_recruit} is $\Theta(n \ln n)$ rounds in the worst case.
\end{proof}

\subsubsection{With Idle Ants}

Here we prove the time complexity of Algorithm \ref{alg_recruit} when there are idle ants in the colony,
assuming $\beta_i > 0$ for all $i$.
If the idle factors are constant, such that $\forall i: \beta_i = O(1)$,
Corollary \ref{cor_slack_ln} shows that the expected time complexity of the algorithm is $O(\ln n)$ in the worst case.

\begin{theorem}
	\label{theorem_alg_lazy_runtime}
	When utilizing idle ants, such that $\forall i: \beta_i > 0$,
	the expected time complexity of Algorithm \ref{alg_recruit} is
	$O(\ln{n} \sum_{i=2}^t \frac{1}{\beta_i^R})$ rounds in the worst case.
\end{theorem}

\begin{proof}
	The proof is similar to Lemma \ref{lemma_alg_upper}; some details are omitted to avoid repetition.
	
	Let $X^0, D$ be an initial assignment and demand, and assume $th=R$ and $\forall x: \beta_x > 0$.
	If the assignment does not meet the demand, then there must be some task $i<t$ which exceeds the demand. 
	We show that after $\frac{\ln n}{\beta_{i+1}^R}$ rounds, the expected number of not-needed ants
	remaining in task $i$ is at most one.
	In the worst case, this happens sequentially one task after the other,
	thus after $O(\ln{n} \sum_{i=2}^t \frac{1}{\beta_i^R})$ rounds no task (except the last) exceeds its demand.
	
	Let $k_i^r = x_i^r - d_i$ denote the number of not-needed ants in task $i$ at round $r$.
	The number of recruiting (i.e., working or idle) ants in task $i+1$ is at least $\beta_{i+1} n$, since there are $\beta_{i+1} n$ idle ants in task $i+1$;
	thus for each not-needed ant in $i$,
	the probability to interact with a sufficient number of recruiting ants in task $i+1$,
	and thus be recruited to task $i+1$, is at least $\beta_{i+1}^R$.
	The probability of an ant to \emph{not} be recruited is thus at most $1 - \beta_{i+1}^R$.
	Let $x$ be some number of rounds, the probability for an ant to remain in task $i$ after $x$ rounds is at most $(1 - \beta_{i+1}^R)^x$.
	The expected number of ants remaining in task $i$ after $x$ rounds is thus at most $k_i^r (1 - \beta_{i+1}^R)^x$.
	Assigning $x = \frac{\ln n}{\beta_{i+1}^R}$, after $\frac{\ln n}{\beta_{i+1}^R}$ rounds the expected number of ants remaining in task $i$ is thus at most:
	$$ k_i^r \cdot (1 - \beta_{i+1}^R)^{\frac{\ln n}{\beta_{i+1}^R}} \leq k_i^r \cdot (\frac{1}{e})^{\ln n} = \frac{k_i^r}{n} \leq 1 $$
	
	Thus, after $\frac{\ln n}{\beta_{i+1}^R}$ rounds, the expected number of not-needed ants
	remaining in task $i$ is at most one
	(which switches after an expected $\frac{1}{\beta_{i+1}^R}$ rounds in the worst case),
	and the overall upper bound is $O(\ln{n} \sum_{i=2}^t \frac{1}{\beta_i^R})$.
\end{proof}

Theorem \ref{theorem_alg_lazy_runtime} directly brings us to the desired runtime
when the idle factors are constant,
proving that the algorithm terminates in an expected $O(\ln{n})$ rounds.

\begin{corollary}
	\label{cor_slack_ln}
	If the idle factors are non-zero constant fractions, such that $\forall i$: $0 < \beta_i \leq 1$ and $\beta_i \in O(1)$,
	the expected time complexity of Algorithm \ref{alg_recruit} is $O(\ln{n})$ rounds in the worst case.
\end{corollary}

\section{Lower Bound}
\label{section_bounds}

We now present a lower bound of $\Omega(n)$ on any algorithm that solves
task allocation, thus showing the exponential gap between
the algorithm time complexity with idle ants, and the lower bound without idle ants.

\begin{theorem}
	\label{theorem_lowerbound_thresh}
	Any protocol that solves task allocation without utilizing idle ants
	requires at least an expected $\Omega(n)$ rounds to complete successfully in the worst case.
\end{theorem}

\begin{proof}
	We prove the theorem by providing an example that requires at least $O(n)$ rounds to terminate successfully.
	Let us define the following assignment and demand vectors for $t \geq 3$ tasks:
	$$ X^0 = \{ 2, 1, \dots, 1, n-t \} ;
	D=\{ 1, 1, \dots, 1, n-t+1 \} $$
	
	In the above assignment, task $1$ exceeds the demand by one ant,
	and task $t$ is short of the demand by one ant.
	All other tasks exactly meet their demand.
	To satisfy the demand, an ant must switch from task $1$ to task $2$, then an ant must switch from task $2$ to task $3$,
	and so on until the demand $D$ is satisfied.
	
	Task $2$ contains a single ant;
	in a single interaction, the probability of any ant to interact with a specific ant is $\frac{1}{n}$,
	thus the probability of any ant to interact with the ant in task $2$ in a single interaction is $\frac{1}{n}$.
	Thus, an ant interacts with the ant in task $2$ after an expected $n$ total interactions.
	There are two ants in task $1$, each performing $R$ interactions per round,
	thus an ant in task $1$ interacts with an ant in task $2$ after an expected $\frac{n}{2 R} = O(n)$ rounds.
	
	Since in the ANTTA model an ant must interact with an ant in the next task in order to switch tasks,
	and task switching is done one-way,	an ant in task $1$ must interact with an ant in task $2$ in order to switch to task $2$,
	and for the demand $D$ to be met, thus the expected lower bound for any protocol solving ANTTA
	is $\Omega(n)$ rounds in the worst case.	
\end{proof}

\section{Idle Distribution}
\label{section_idle}

Idle ants exponentially reduce the time it takes the colony to meet certain demands,
as shown in Section \ref{section_algorithm}.
However,  idle ants do not exist in equal amounts in each stage of the colony's life;
biologists observed that new nests start with very few or no idle ants, and as the nest matures, the percentage of idle ants increases \cite{GordonBook,gordon1989dynamics}.
This is consistent with the possibility that each ant has certain decision points in which it might become idle with some probability.
Thus, a larger percentage of the ants become idle, until reaching a stable fraction of the nest.

Young nests contain very few or no idle ants, and thus have a slow response time to sudden changes in demands for different tasks.
When encountering such demands, the nest may be incentivized to turn more ants into idle, in order to improve the response time of task allocation.
However, these idle ants, once becoming idle, need to be distributed among the tasks to which they recruit to enable the fast solution of task allocation.

In this section we present the problem of \emph{Idle Distribution}, in which idle ants determine which task they recruit to,
where the goal is to reach a state in which a constant fraction of idle ants are in each task,
We show a lower bound of $\Omega(n)$ rounds to solve idle distribution.
Thus, when a new demand is encountered for the first time, young nests still have to solve the idle distribution,
and are thus still slow to respond, providing further explanation for their low survivability \cite{LowSurvive,GordonBook}.


\subsection{Model Extensions}

Denote by $m$ the number of idle ants, such that $\sum_i \beta_i n = m$.
Here we assume $t \leq m \leq n$.

For an \emph{idle} ant $a$, we denote by $T_a^r$ the task into which it recruits other ants; the idle ant $a$ does not actually belong (i.e., contribute) to that task.
For simplicity, we assume that when an ant becomes idle, it recruits to whatever task it was assigned previously.

Idle ants switch tasks \emph{only} through interactions and recruitment,
following the biological model and observations \cite{AdditionalGordon,Pacala};
in order to switch to task $i$ in round $r$, an idle ant must interact with an ant in task $i$ in round $r-1$,
i.e., for an \emph{idle} ant $a$, if $T_a^{r+1}=i$ then it holds that either $T_a^r=i$ or $i \in \{ T_{b_1}^r, \dots, T_{b_R}^r \}$,
where $\{ b_1,\dots,b_R \}$ is the set of ants which $a$ interacts with in round $r$.

In the \emph{idle distribution problem}, an adversary determines the number of idle ants $m$, the initial allocation to tasks of the idle ants,
and the initial assignment of non-idle ants $X^0$.
The goal of the protocol is to distribute the idle ants among the tasks, such that each task will contain a constant fraction of the idle ants.
A protocol terminates successfully in the first round $f$ in which the idle factor of each task is a constant fraction of $m$, such that $\forall i: 0 < \beta_i \leq 1, \beta_i = O(1)$.
When $m$ is a constant fraction of $n$ (the total number of ants in the colony),
the resulting idle factors of each task are non-zero constant fractions of $n$, as assumed by Corollary \ref{cor_slack_ln}.

\subsection{Lower Bound}

\begin{theorem}
	Any protocol that solves Idle Distribution requires at least an expected $\Omega(n)$ rounds to complete successfully in the worst case.
\end{theorem}

\begin{proof}
	We will prove the theorem by an example that requires at least $\Omega(n)$ rounds to terminate successfully.
	Let $t = 2$ be the number of tasks, and let $m = 2$,
	such that, to terminate successfully, a single idle ant must belong to each task.
	Let us define the following assignment vector: $X^0 = \{n-3,1\}$.
	In addition, both idle ants start in task $1$.
	
	To recruit to task $2$, an idle ant must interact with an ant in task $2$.
	There is only one ant in task $2$, thus, each idle ant has a $\frac{R}{n}$ probability to interact with it.
	Since there is a constant number of idle ants, it takes at least an expected $O(\frac{n}{R}) = O(n)$ rounds.
\end{proof}

\section{Discussion}
\label{section_discussion}
In large enough colonies, such as a nest of $5000$ ants
(a standard size of a mature colony in various ant species \cite{GordonBook}),
a runtime of $O(n)$ is unacceptable:
with $R=1$ and round length of $1$ second, it would take a nest \emph{over an hour} to meet the demand.
In recent years, biological observations \cite{lazybones,WhyLazy,LazyInNature} discovered that,
in many species of ants, a substantial number (20\%-50\%) of the ants in the colony
are idle, even when there is work to be done.

We have presented a new model for the task allocation problem, inspired by the biology of ant colonies.
Under this model, we devised an algorithm that mimics the behavior of ant colonies in order to solve the task allocation problem.
The time complexity of the algorithm is $O(\ln n)$ rounds when the number of idle ants constitutes
a constant fraction of the number of ants in the colony (the time complexity of the algorithm without idle ants is $O(n \ln{n})$ rounds).
On the other hand, when assuming there are no idle ants,
the lower bound is $\Omega(n)$ rounds which is exponentially worse.
This requirement for a constant fraction of the colony to be idle ants is actually portrayed in nature.
Biologists observed that, in ant colonies, approximately 20-50\% of the worker ants
are idle \cite{lazybones,WhyLazy,LazyInNature}.
The gap described in the previous paragraph provides a possible explanation for this phenomenon
which, to the best of our knowledge, is still unexplained by biologists.

In addition, we presented a model for distributing the idle ants among the tasks,
such that enough idle ants recruit to each task, as assumed by the gap shown.
In some species, the survival rate of young nests is very low;
only 10\% of new nests survive through the first year \cite{LowSurvive,GordonBook}.
Our model gives a possible explanation for  this phenomenon which is also, to the best of our knowledge, as of yet unexplained.
Ants in new ant colonies, likely to encounter sudden changes in demands for different tasks, might be incentivized to become idle to improve task allocation.
Our model shows a lower bound of $\Omega(n)$ for the process of idle distribution,
showing why the total time to do both task allocation and idle distribution is much longer
than just task allocation when there are already enough idle ants,
thus providing a possible explanation for the low survivability of new ant colonies.

\subsection{Future Work}

Following our work, several additional topics arise that may be of interest:

An interesting topic is that of density, the probability of interacting with another ant.
Our model assumes a completely uniform density, such that interactions are uniformly random,
i.e., each ant has an equal probability to interact with each other ant.
Consider what would happen if an ant interacts with higher probability with some subset of ants,
such as ants of the same task, ants previously interacted with, or even according to the physical location of the ants.


Further topics for research:
\begin{itemize}
	\item
	Interactions with ants of different nests,
	each nest aiming to achieve its own demands without interference,
	perhaps even with competing demands.
	
	\item
	Differentiation among ants, such that each ant may contribute different work to each task,
	or may even be excluded from certain tasks completely.
	
	\item
	Incorporating work or success rate, such that the \emph{working} and \emph{not-needed} states are assigned probabilistically,
	according to the assignment and demand of each task.
\end{itemize}

\section*{Acknowledgements}
The authors thank Hanoch Levy, Roman Kecher, and Uri Zwick for very helpful discussions.

\clearpage
\nocite{*}
\bibliographystyle{abbrv}

\end{document}